\documentclass[US-letter,USenglish,authorcolumns,autoref]{lipics-v2021}
\usepackage[noend]{algpseudocode}
\usepackage{algorithm}
\usepackage{times}
\usepackage{amssymb}
\usepackage{amsmath}
\usepackage{latexsym}
\usepackage{graphics}
\usepackage{url}
\usepackage{verbatim}
\usepackage{color}
\usepackage{setspace}
\usepackage{multirow}
\usepackage{lineno}
\usepackage{booktabs}
\usepackage{graphicx}
\usepackage{caption}
\usepackage[shortlabels]{enumitem}

\renewcommand{\tt} {\mathtt}

\definecolor{winered}{rgb}{0.5,0.2,0}
\definecolor{lipicsYellow}{rgb}{0.99,0.78,0.07}

\usepackage{hyperref}
\hypersetup{ 
    colorlinks = true,
    allcolors={winered},
}

\usepackage{wrapfig}
\usepackage[normalem]{ulem}
\usepackage{changepage}

\algnewcommand\And{\textbf{and}}
\algnewcommand\Or{\textbf{or}}
\algnewcommand\Not{\textbf{not}}
\algnewcommand\In{\textbf{in}}
\algnewcommand\Each{\textbf{each}}

\newtheorem{property}[theorem]{Property} 

\newcommand{\squishlist}{
 \begin{list}{$\bullet$}
  { \setlength{\itemsep}{0pt}
     \setlength{\parsep}{3pt}
     \setlength{\topsep}{3pt}
     \setlength{\partopsep}{0pt}
     \setlength{\leftmargin}{2.5em}
     \setlength{\labelwidth}{1em}
     \setlength{\labelsep}{0.5em} } }

\newcommand{\squishlisttwo}{
 \begin{list}{$\triangleright$}
  { \setlength{\itemsep}{0pt}
     \setlength{\parsep}{0pt}
    \setlength{\topsep}{0pt}
    \setlength{\partopsep}{0pt}
    \setlength{\leftmargin}{2em}
    \setlength{\labelwidth}{1.5em}
    \setlength{\labelsep}{0.5em} } }

\newcommand{\squishend}{
  \end{list}  }

\usepackage{listings}

\definecolor{verbgray}{gray}{0.9}

\lstnewenvironment{code}{%
  \lstset{backgroundcolor=\color{verbgray},
 frame=single,
  language=C,
  framerule=0pt,
  basicstyle=\ttfamily,
  showstringspaces=false,
  commentstyle=\color{blue}\textit,
  keywordstyle=\color{black}\bf,
  numbers=none,
  keepspaces=true,
  columns=fullflexible}}{}

\newcommand{\bS}{\mathbf{S}}

\newcommand{\bN}{\mathbf{N}}

\newcommand{\bp}{\mathbf{p}}

\newcommand{\bs}{\mathbf{s}}

\newcommand{\br}{\mathbf{r}}
\newcommand{\bq}{\mathbf{q}}
\newcommand{\bt}{\mathbf{t}}

\definecolor{shadecolor}{rgb}{.91, .91, .91}
\definecolor{bordercolor}{rgb}{.8, .8, .6}

\definecolor{ultramarine}{rgb}{0, 0.125, 0.376}

 \definecolor{arsenic}{rgb}{0.23, 0.27, 0.29}
 \definecolor{beige}{rgb}{0.96, 0.96, 0.86}

\definecolor{amber}{rgb}{1.0, 0.75, 0.0}
\definecolor{orange}{rgb}{1.0, 0.49, 0.0}
\definecolor{dandelion}{rgb}{0.94, 0.88, 0.19}

  \definecolor{indiagreen}{rgb}{0.07, 0.53, 0.03}
  \definecolor{huntergreen}{rgb}{0.21, 0.37, 0.23}

\newcommand{\blue}[1] {\textcolor{blue}{#1}}
\newcommand{\red}[1] {\textcolor{red}{#1}}

\newcommand{\defo}[1] {\emph{\textcolor{blue}{#1}}}

\definecolor{shadecolor}{rgb}{.9, .9, .9}

 \usepackage{framed}

    {\endMakeFramed}

    \newenvironment{frshaded*}{%
    \MakeFramed {\advance\hsize-\width \FrameRestore}}%
    {\endMakeFramed}
    
    \newcounter{examplecounter}
\newenvironment{exam}{
 \begin{frshaded*}
    \refstepcounter{examplecounter}%
    \noindent
  \textbf{Example \arabic{examplecounter}}%
  \quad
}{%
\end{frshaded*}
}

{\endMakeFramed}

\newenvironment{frshaded2*}{%
    \MakeFramed {\advance\hsize-\width \FrameRestore}}%
    {\endMakeFramed}

\newenvironment{result}{
 \begin{frshaded2*}
}{%
\end{frshaded2*}

}

{\endMakeFramed}

\newenvironment{frshaded3*}{%
    \MakeFramed {\advance\hsize-\width \FrameRestore}}%
    {\endMakeFramed}

\graphicspath{{graphics/}}

\DeclareMathOperator{\wt}{weight}
\DeclareMathOperator{\ap}{ap}
\DeclareMathOperator{\neck}{neck}
\DeclareMathOperator{\rank}{rank}

\DeclareMathOperator{\comp}{comp}

\nolinenumbers

\title{Decoding universal cycles for $t$-subsets and $t$-multisets by decoding bounded-weight de Bruijn sequences}
\titlerunning{~~Decoding universal cycles}
\author{Daniel Gabri\'{c}}{University of Guelph, Canada}{}{}{}
\author{Wazed Imam}{University of Guelph, Canada}{}{}{}
\author{Lukas Janik-Jones}{University of Guelph, Canada}{}{}{}
\author{Joe Sawada}{University of Guelph, Canada}{}{}{}

\author{~}{~}{}{}{}
\authorrunning{~}

\Copyright{~}
\begin{document}

%-----------  TITLE STUFF ---------------

%\title{ Decoding universal cycles for the $k$-subsets of an $n$-set}
%\author{ D. Gabric, W. Imam, L. Janik-Jones, J. Sawada \\
%}
\maketitle

\frenchspacing

\begin{abstract}  \normalsize
A universal cycle for a set $\mathbf{S}$ of combinatorial objects is a cyclic sequence of length $|\mathbf{S}|$ that contains a \emph{representative} of each element in $\mathbf{S}$ exactly once as a substring.
Despite the many universal cycle constructions known in the literature for various sets including $k$-ary strings of length $n$, permutations of order $n$, $t$-subsets of an $n$-set, and $t$-multisets of an $n$-set, remarkably few have efficient decoding (ranking/unranking) algorithms.  In this paper we develop the first polynomial time/space decoding algorithms for bounded-weight de Bruijn sequences for strings of length $n$ over an alphabet of size $k$.  %The ranking \red{(do we want to switch back and forth between ranking and encoding, or do we want to stick with one?)} algorithms require $O(n^3k^2)$ time and the unranking algorithms require $\mathcal{O}(n^4k^2 \log k)$ time, both using polynomial space.  
The results are then applied to decode universal cycles for $t$-subsets and $t$-multisets. 
\end{abstract}

%=====================================================================
%=====================================================================
%=====================================================================
\section{Introduction}

A \defo{universal cycle} for a set $\mathbf{S}$ of combinatorial objects is a cyclic sequence of length $|\mathbf{S}|$ that contains a \emph{representation} of each element in $\mathbf{S}$ exactly once as a substring.  Unless otherwise stated, assume that the representative of each element $\bs \in \bS$ is simply $\bs$ itself.  In cases where $\bs$ is not a string, then its representative will be clearly articulated.  The following ``fundamental questions'' can be asked of any set $\bS$~\cite{chung}:
\smallskip

\begin{adjustwidth}{0.5cm}{0.3cm}
\begin{enumerate}
    \item[(i)] Do universal cycles \emph{exist} for $\bS$?
    \item[(ii)] How many universal cycles for $\bS$ are there?
    \item[(iii)] Can a universal cycle for $\bS$ be  (efficiently) constructed?
    \item[(iv)] Does a universal cycle for $\bS$ exist with properties that make it efficiently \emph{decodable}? In other words, are there efficient ranking and ranking algorithms for a specific universal cycle of $\bS$?
\end{enumerate}
\end{adjustwidth}

\smallskip
\noindent
There are numerous results pertaining to the first three questions for a wide
variety of interesting sets $\bS$; see for instance~\cite{fred-nfsr,chung,2009research,godbole2011,etzion-book,dbseq}. The decoding question, though, has proven to be more elusive.  

Let $\bS_k(n)$ denote the set of all length-$n$ strings over $\{\tt{1,2,\ldots, k}\}$.  A universal cycle for $\bS_k(n)$ is known as a \defo{de Bruijn sequence}.  The \defo{weight} of a string is the sum of its symbols. Let $\bS_k(n,w^\uparrow)$ denote the subset of $\bS_k(n)$ containing the strings with weight \emph{at least} $w$.  Let $\bS_k(n,w_\downarrow)$ denote the subset of $\bS_k(n)$ containing the strings with weight \emph{at most} $w$.  Universal cycles for $\bS_k(n,w^\uparrow)$ and $\bS_k(n,w_\downarrow)$ are known as a \defo{bounded-weight de Bruijn sequences}. 
In this paper we present efficient ranking and unranking algorithms for a bounded-weight de Bruin sequence.  The algorithms are applied to efficiently decode universal cycles for $t$-subsets and $t$-multisets.  They are the first known efficient algorithms to decode specific universal cycles for these objects.

\medskip
%================================
%================================
\noindent{\bf Decoding universal cycles.}~~
Let $\mathcal{U}=u_1u_2\cdots u_m$ be a universal cycle for a set $\mathbf{S}$ of length-$n$ strings.
The \defo{rank} of a string $\bs \in \bS$ is the starting index where $\bs$ is found in $\mathcal{U}$.  If $\bs$ is a prefix of $\mathcal{U}$ then it has rank 1; if $\bs$ starts at the last symbol of $\mathcal{U}$ and wraps around, then it has rank $|\mathbf{S}|$. Given a rank $r$, the \defo{unranking} process determines the string $\bs \in \bS$ starting at index $r$ in $\mathcal{U}$.  Together, ranking and unranking algorithms allow for a universal cycle to be \emph{decoded}.   Any universal cycle can be decoded by traversing the sequence until a given string, or rank, is discovered; however, this requires $O(|{\bf S}|)$ time.  Similarly, a look-up table can be applied to store the position of a given string, but this requires $\Theta(|{\bf S}|)$  space.  We say a decoding algorithm is \defo{efficient} if the ranking and unranking algorithms run in polynomial time and require polynomial space, with respect to the size of the alphabet and the length of each string $\bs$. %\red{do we want to define efficient in terms of the size of the strings, or in terms of the size of the set $\bS$, e.g., a decoding algorithm is efficient if the ranking and unranking can be done in $o(|\bS|)$ time~\cite{Wilf:1982} Then we could expand on what we want, which is even stricter. }
The efficient decoding of a universal cycle is necessary when it is applied to robotic position sensing (vision) applications~\cite{BM,magic}.

 There are numerous de Bruijn sequence constructions that apply a variety of different algorithmic approaches; see, for instance~\cite{fred-nfsr,etzion-book,dbseq}.
However, somewhat surprisingly, only the lexicographically smallest de Bruijn sequence for $\bS_k(n)$ is known to be efficiently decodable for all $n$ and~$k$~\cite{kociumaka,ranking}. We review the decoding approach for this de Bruijn sequence in Section~\ref{sec:granddaddy}.
%
%By applying Lempel's $D$-morphism~\cite{lempel}, 
A recursively constructed de Bruijn sequence can be efficiently decoded in the binary case, but not for all $n$ and $k$~\cite{mitchell-decode,tuliani-decode}. See also~\cite[p.317, ex.97-99]{knuth}.
The only other known efficient decoding algorithms are for universal cycles of permutations using a shorthand representation~\cite{shorthand,shorthand2}.  This highlights the difficulty in discovering such decoding algorithms.

\medskip
%================================
%================================
\noindent{\bf $t$-subsets and $t$-multisets.}~~
Let $\mathbf{Subset}(n,t)$ denote the set of all $t$-subsets of $\{\tt{1,2,\ldots, n}\}$.
Universal cycles for $\mathbf{Subset}(n,t)$  have been considered using a standard string representation for each subset, where, for instance, either $\tt{12}$ or $\tt{21}$ can be used to represent the subset $\{1,2\}$~\cite{chung,jackson,hurlbert,2009research,rudoy}.  Using this representation, universal cycles for $t$-subsets exist only if $n$ divides ${n \choose t}$.  A $t$-subset can also be represented by a length-$n$ binary string with $t$ ones. A shorthand representation excludes the final redundant bit. Applying this shorthand representation, universal cycles for $t$-subsets can be efficiently constructed for all $1 \leq t \leq n$~\cite{fixedweight2}. No efficient decoding algorithms are known for these universal cycles.  
In~\cite{cantwell}, a labelled graph is applied to represent a $t$-subset. With this representation they demonstrate the existence of a universal cycle for $t$-subsets, however, no efficient construction is provided.
A difference representation was recently described in~\cite{difference}.  Given a subset
$\{\tt{s}_1, \tt{s}_2, \ldots , \tt{s}_t\}$ with elements listed in increasing order, its \defo{difference representative} is the length-$t$ string where the first symbol is $\tt{s}_1$ and the $j$-th symbol is  $\tt{s}_j - \tt{s}_{j-1}$ for $1 < j \leq t$. For example, the difference representatives for 
 $\mathbf{Subset}(5,3)$ = \[ \{ \tt{\{1,2,3\}, \{1,2,4\}, \{1,2,5\}, \{1,3,4\}, \{1,3,5\}, \{1,4,5\}, \{2,3,4\}, \{2,3,5\}, \{2,4,5\}, \{3,4,5\}} \}\] 
 are given by
$\tt{111, 112, 113, 121, 122, 131, 211, 212, 221},$ and $\tt{311}$, respectively.  In general, the difference representatives for $\mathbf{Subset}(n,t)$ correspond to the strings in $\bS_{n-t+1}(t,n_\downarrow)$.

%====================
Let $\mathbf{Multiset}(n,t)$ denote the set of all $t$-multisets of $ \{\tt{0,1,\ldots, n-1}\}$.  
Universal cycles for $\mathbf{Multiset}(n,t)$  have been considered using a standard string representation for each subset, where, for instance, either $\tt{112}$, $\tt{121}$, or $\tt{211}$ can be used to represent the multiset $\{1,1,2\}$~\cite{jackson,hurlbert2009}.  Using this representation, universal cycles for $t$-multisets exist only if $n$ divides ${n+t-1 \choose t}$.  Like with $t$-subsets, 
we can also apply a difference representation for $t$-multisets~\cite{difference}.   Given a multiset
$\{\tt{m}_1, \tt{m}_2, \ldots , \tt{m}_t\}$ with elements listed in non-decreasing order, its \defo{difference representative} is the length $t$ string where the first symbol is $\tt{m}_1$ and the $j$-th symbol is  $\tt{m}_j - \tt{m}_{j-1}$ for $1 < j \leq t$.
For example, the difference representatives for 
 $\mathbf{Subset}(3,3)$ = \[ \{ \tt{\{0,0,0\}, \{0,0,1\}, \{0,0,2\}, \{0,1,1\}, \{0,1,2\}, \{0,2,2\}, \{1,1,1\}, \{1,1,2\}, \{1,2,2\}, \{2,2,2\}} \}\] 
 are given by
$\tt{000, 001, 002, 010, 011, 020, 100, 101, 110},$ and $\tt{200}$, respectively.  By replacing each symbol $\tt{x}$ with $\tt{x}+1$ in the difference representatives for $\mathbf{Multiset}(n,t)$, the resulting set of strings is $\bS_{n}(t,(n{+}t{-}1)_\downarrow)$.

\bigskip

%=====================================================================
\noindent
{\bf Main results.} 
The main results of this paper are as follows:

\begin{enumerate}
    \item We demonstrate that the lexicographically smallest universal cycle for $\bS_k(n,w^\uparrow)$ (a bounded-weight de Bruijn sequence) is efficiently decodable. 
    \item We demonstrate a universal cycle for $t$-subsets that is efficiently decodable.
    \item We demonstrate a universal cycle for $t$-multisets that is efficiently decodable.
   % \item We efficiently rank and unrank the list of $k$-ary necklaces of length $n$ with weight at least $w$ as they appear in lexicographic order.
\end{enumerate}

%\begin{theorem}
%    The lexicographically smallest bounded-weight de Bruijn sequence for $\bS_k(n,w)$
%    can be ranked in $\mathcal{O}(n^3k^2)$ time and unranked in $\mathcal{O}(n^4k^2 \log k)$ time, using $\mathcal{O}(n^3k)$ space. 
%\end{theorem}

%
\noindent
As we mention later in the paper, a decodable universal cycle for $\bS_k(n,w^\uparrow)$ can be easily applied to decode a universal cycle for the strings in $\bS_k(n,w_\downarrow)$.
The analysis of all algorithms described in this paper assume the unit-cost RAM model of computation. A C implementation of our decoding algorithms is available at \url{https://debruijnsequence.org/db/subset_decode}.

\bigskip

%=====================================================================
\noindent
{\bf Outline.}
The remainder of the paper is organized as follows.  In Section~\ref{sec:back} we present background definitions and notation.  In Section~\ref{sec:granddaddy} we review the efficient decoding approach for the lexicographically smallest de Bruijn sequence. In Section~\ref{sec:bounded} we generalize these results by considering a lower bound on the weight of the strings;  in Section~\ref{sec:upper} we consider an upper bound on the weight.
In Section~\ref{sec:subset}, we apply these results to demonstrate efficiently decodable universal cycles for $t$-subsets and $t$-multisets.  In Section~\ref{sec:rank}, we prove the technical results. 

%In Section~\ref{sec:rank}, we give a brief summary and provide details of future research. 

%which also lead to efficient ranking and unranking algorithms for $k$-ary necklaces of length $n$ with weight at least $w$ as they appear in lexicographic order.  

%=====================================================================
%=====================================================================
%=====================================================================
\section{Preliminaries} \label{sec:back}

Recall that $\bS_k(n)$ denotes the set of all length-$n$ strings over $\{\tt{1,2, \ldots, k}\}$.  Consider two strings $\bs$ and $\bt$.  
Let  $\bs \cdot \bt$ denote the  concatenation of $\bs$ and $\bt$, and   let $\bt^j$ denote $j$ copies of $\bt$ concatenated together.
A \defo{necklace} is the lexicographically smallest string in an equivalence class of strings under rotation.  
Let $\neck(\bs)$ denote the lexicographically smallest rotation of $\bs$, i.e., its necklace. Let $\ap(\bs)$ denote the \defo{aperiodic prefix} of $\bs$, i.e., the shortest string $\bt$ such that $\bs = \bt^j$ for some $j \geq 1$.
For example, $\ap(\tt{1111}) = \tt{1}$.  A string $\bs$ is said to be \defo{aperiodic} (primitive) if $\ap(\bs) = \bs$; otherwise $\bs$ is \defo{periodic} (a power). %A \defo{Lyndon word} is an aperiodic necklace. 

Let $\bN_k(n) = (\sigma_1, \sigma_2, \ldots, \sigma_t)$ denote the tuple consisting of all necklaces in $\bS_k(n)$ listed in lexicographic order. For example, 
\[ \bN_2(4) = ( \tt{1111, 1112, 1122, 1212, 1222, 2222} ). \]
Amazingly, by concatenating together the aperiodic prefixes of the necklaces in $\bN_k(n)$~\cite{fredricksen-beads-colors} we obtain the lexicographically smallest de Bruijn sequence for $\bS_k(n)$~\cite{fredricksen-beads-colors,fred-lex-least}.
Let
\[\mathcal{G}_k(n) = \ap(\sigma_1) \ap(\sigma_2)  \cdots \ap(\sigma_t), \]
denote this lexicographically smallest de Bruijn sequence for $S_k(n)$.
For example,
\[ \mathcal{G}_2(4)  = \tt{1 \cdot 1112 \cdot  1122 \cdot 12 \cdot 1222 \cdot 2}.\]
The sequence $\mathcal{G}_k(n)$ is also known as the \defo{Granddaddy}\footnote{
This name was originally given to Martin's prefer-largest greedy construction~\cite{martin} by Knuth~\cite[p.317]{knuth}.  
It is equivalent to $\mathcal{G}_k(n)$ by a simple mapping of the alphabet symbols.} de Bruijn sequence and the \defo{Ford} sequence\footnote{Based on a 1957 technical report by Lester Ford who independently discovered the sequence via an equivalent greedy construction.}. 

\medskip
\noindent
{\bf A note on notation.} In this paper we use bold lower case letters such as $\br,\bs, \bt$ to denote arbitrary strings, while we reserve Greek letters such as $\alpha, \beta,\sigma, \gamma$ to denote necklaces.

\subsection{Bounding the weight}
The \defo{weight} of a string $\bs$, denoted by $\wt(\bs)$, is the sum of it symbols.  Recall that $\bS_k(n,w^\uparrow)$ denotes the subset of $\bS_k(n)$ containing strings with weight \emph{at least} $w$.
Let $\bN_k(n,w^\uparrow) = (\alpha_1,\alpha_2, \ldots, \alpha_{m})$ denote the tuple consisting of all necklaces in $\bS_k(n,w^\uparrow)$ listed in lexicographic order.  For example, 
\[ \bN_2(4,6^\uparrow) =  ( \tt{1122, 1212, 1222, 2222} ). \]

\noindent
Remarkably, the Granddaddy construction can be generalized to $\bS_k(n,w^\uparrow)$.
Let
\[\mathcal{G}_k(n,w^\uparrow) = \ap(\alpha_1) \ap(\alpha_2)  \cdots \ap(\alpha_{m}).  \] 
The sequence $\mathcal{G}_k(n,w^\uparrow)$ is the lexicographically smallest bounded-weight de Bruijn sequence for $\bS_k(n,w^\uparrow)$~\cite{weight1,generalize-classic-greedy}.
For example,
\[ \mathcal{G}_2(4,6^\uparrow)  = \tt{ 1122 \cdot 12 \cdot 1222 \cdot 2}.\]

\noindent
Unfortunately, adapting the Granddaddy for strings with an upper bound on the weight does not necessarily lead to a universal cycle.  For example,  the sequence $\tt{1 \cdot 1112 \cdot 1122 \cdot 12}$ is not a universal cycle for $\bS_2(4,6_\downarrow)$: the string $\tt{2211}$ is missing.

There are a number of other known constructions for bounded-weight de Bruin sequences (for instance, see~\cite{weight1,generalize-classic-greedy,karyframework,concattree}); however, no efficient decoding algorithms for them was previously known. Efficient universal cycles are also known for the subsets of $\bS_k(n)$ with both a lower and upper bound on weights~\cite{sawada2013}.  

\subsection{Necklace properties}

The following five properties of necklaces can be derived from the definition of a necklace.

\begin{property} \label{prop:lastnon}
If $\alpha = \tt{a}_1\cdots \tt{a}_{j-1} \tt{x} \tt{k}^{n-j}$ is in $\bN_k(n)$ for some $j\geq 1$ and symbol $\tt{x} \neq \tt{k}$, then $\tt{a}_1\cdots \tt{a}_j (\tt{x}+1) \tt{k}^{n-j}$ is also a necklace.
\end{property}

\begin{property} \label{prop:first}
Let $n,k > 1$ and $k < w < kn$.  The first necklace in $\bN_k(n,w^\uparrow)$ is $\tt{1}^{n-j-1}\tt{x}\tt{k}^{j}$,
where $j=\lfloor \frac{w-n}{k-1}\rfloor$ and $\tt{x} = w-(n{-}j{-}1)-kj$; it has weight $w$ and is aperiodic.
\end{property}
\begin{property} \label{prop:last}
 Let $n,k > 1$ and $w < kn$.  The last two necklaces in  $\bN_k(n,w^\uparrow)$ are
 $(\tt{k{-}1})\tt{k}^{n-1}$ and $\tt{k}^n$.
\end{property}

\noindent
The next property follows from the definition of a periodic necklace.

\begin{property} \label{prop:periodic}
Let $\alpha_i$ and $\alpha_{i+1}$ be consecutive necklaces in $\bN_k(n,w^\uparrow)$ where $n,k > 1$ and $w < kn$.   If $\alpha_i$ is periodic then $\ap(\alpha_i) \cdot \ap(\alpha_{i+1})$ has prefix $\alpha_i$.  Moreover $\alpha_{i+1}$ is aperiodic.
\end{property}
\begin{proof}
Since $\alpha_i$ is periodic, it can be written as $\alpha_i = (\tt{a}_1\cdots \tt{a}_p)^j$ for some $j > 1$.  Let $\ell$ be the largest index in $[1,p]$ such $\tt{a}_\ell \neq \tt{k}$.  Such an index exists since $\alpha_i$ is not the last necklace $\tt{k}^n$ in $\bN_k(n,w^\uparrow)$. By Property~\ref{prop:lastnon}, there exists a necklace with prefix $(\tt{a}_1\cdots \tt{a}_p)^{j-1}$ that is lexicographically larger than $\alpha_i$ and also in $\bN_k(n,w^\uparrow)$. Thus $\alpha_{i+1}$ must also have prefix  $(\tt{a}_1\cdots \tt{a}_p)^{j-1}$ since $\bN_k(n,w^\uparrow)$ is lexicographically ordered.  Thus $\ap(\alpha_i) \cdot \ap(\alpha_{i+1})$ has prefix $\alpha_i$. The fact that $\alpha_{i+1}$ is aperiodic is proved in~\cite[Lemma 5]{generalize-classic-greedy}.
\end{proof}

\noindent The following property is proved in~\cite{ranking}.

 \begin{property} \label{prop:suffix}
Let $\alpha = \tt{a}_1\tt{a}_2\cdots \tt{a}_n$ be a necklace and let $1 \leq i \leq  j \leq n$.
Then $\tt{a}_i \tt{a}_{i+1}\cdots a_{\tt{j}-1} \tt{x}$ is lexicographically larger than  $\alpha$ if $\tt{x} > \tt{a}_j$.  
\end{property}

%=====================================================================
%=====================================================================
%=====================================================================
\section{Ranking the Granddaddy de Bruijn sequence $\mathcal{G}_k(n)$} \label{sec:granddaddy}

In this section we review how the Granddaddy de Bruijn sequence $\mathcal{G}_k(n)$ can be efficiently decoded based on the groundbreaking results of Kociumaka,  Radoszewski, and  Rytter~\cite{kociumaka}.  Then in Section~\ref{sec:bounded}, we adapt the approach to decode the bounded-weight de Bruijn sequence $\mathcal{G}_k(n,w^\uparrow)$.

Let $\bs$ be a string in $\bS_k(n)$. It can be written as $\bs = \bp\bq$ where $\bq$ is the longest suffix of $\bs$ such that $\bq\bp$ is a necklace.  For example, if $\bs = \tt{221221}$
then $\bp = \tt{22}$ and $\bq = \tt{1221}$.
Let $\rank(\bs)$ denote the rank of $\bs$ in $\mathcal{G}_k(n)$. 
The last $n$ length-$n$ substrings of $\mathcal{G}_k(n)$ are all strings $\bs = \tt{k}^{n-j}\tt{1}^j$ for $0  \leq  j < n$. For these strings $\rank(\bs) = k^n - (n-j)+1$. 
For all other strings, there are two key steps involved in the efficient computation of $\rank(\bs)$. 

\medskip 
\noindent
{\bf Step 1:} Determine consecutive necklaces $\beta_1$, $\beta_2$, and $\beta_3$ in $\bN_k(n)$ such that $\ap(\beta_1)\ap(\beta_2)\ap(\beta_3)$ contains $\bs$ as a substring. 
The following remark summarizes results from~\cite{fredricksen-beads-colors}.

\begin{remark} \label{rem:granddaddy}
     There exist unique consecutive necklaces $\beta_1$, $\beta_2$, and $\beta_3$ in $\bN_k(n)$ (considered cyclically) such that 
    $\bp$ is a suffix of $\beta_1$ and $\bq$ is a prefix of $\beta_2$, i.e., $\bs$ is a substring of $\beta_1\beta_2$.  Moreover, $\bs$ is a substring of $\ap(\beta_1)\ap(\beta_2)\ap(\beta_3)$.
    %, noting $\beta_3$ may be needed only if $\beta_2$ is periodic.
\end{remark}
%  
%In the above remark, if $\bs$ is a necklace, $\bq = \beta_2 = \bs$.  
The necklaces $\beta_1,\beta_2$, and $\beta_3$ from the above remark can be determined in $\mathcal{O}(n^2)$ time using $\mathcal{O}(n)$ space~\cite{kociumaka}.

\medskip

\noindent
{\bf Step 2:} Compute $T_k(n,\sigma_i) = |\ap(\sigma_1)\ap(\sigma_2)\cdots \ap(\sigma_{i-1})|$ which counts the number of strings in $\bS_k(n)$ whose necklaces are lexicographically smaller than $\sigma_i$.  For example, 
\[ T_2(4,\tt{1122}) = |\{\tt{1111}\}| + |\{\tt{1112, 1121, 1211, 2111}\}| = 5. \]
The value $T_k(n,\sigma_i)$ can be computed in $\mathcal{O}(n^2)$ time  using $\mathcal{O}(n^2)$ space~\cite{kociumaka}.

\medskip

Putting these two steps together,

\begin{center}
$\rank(\bs) = \left\{ \begin{array}{ll}
         k^n-(n{-}j)+1,    \    &\ \  \mbox{if $\bs = \tt{k}^{n-j}\tt{1}^j$ for some $0 \leq j < n$;}\\
         T_k(n, \beta_2), - |\bp|+1, \  &\ \  \mbox{otherwise.} 
         \end{array} \right. $
\end{center}
\begin{exam}
Recall $\mathcal{G}_2(4) = \tt{1  111\underline{2   112}2  12  1222  2}$. Consider $\bs = \tt{2112}$. It is a substring in the concatenation of $\beta_1 = \tt{1112}$ and $\beta_2 = \tt{1122}$, where $\bp = \tt{2}$ and $\bq = \tt{112}$.  Since $T_2(4,\tt{1122}) = 5$, the rank of $\bs$ is $5 - 1 + 1 = 5$.  The rank of $\tt{2211}$ is $2^4 - 2 + 1 = 15$.
\end{exam}
\noindent
The string at rank $r$ in $\mathcal{G}_k(n)$ can be determined in $\mathcal{O}(n^3 \log k)$ time by applying a binary search and repeatedly applying the ranking algorithm~\cite{kociumaka}.  See also a reinterpretation of these results in~\cite{ranking}.

%=====================================================================
%=====================================================================
%=====================================================================
\section{Decoding the bounded-weight de Bruijn sequence $\mathcal{G}_k(n,w^\uparrow)$} \label{sec:bounded}

In this section we adapt the decoding algorithm for the Granddaddy to decode the bounded-weight de Bruijn sequence $\mathcal{G}_k(n,w^\uparrow)$.   If $w=k$, then $\mathcal{G}_k(n,w^\uparrow) = \mathcal{G}_k(n)$. If $w=kn$, then $\mathcal{G}_k(n,w^\uparrow)$ is the single symbol $\tt{k}$.   
%Recall that $\bN_k(n,w^\uparrow) = (\alpha_1,\alpha_2, \ldots, \alpha_{m})$ is the tuple consisting of all necklaces in $\bS_k(n,w^\uparrow)$ listed in lexicographic order.  
Throughout the remainder of this section assume that $k < w < kn$ and let:
\begin{itemize}
    \item $\alpha_1 = \tt{a}_1\tt{a}_2\cdots \tt{a}_n$ be the first necklace in $\bN_k(n,w^\uparrow)$, 
    \item $\bs = \bp\bq$ be a string in $\bS_k(n,w^\uparrow)$, where $\bq$ is the longest suffix of $\bs$ such that $\bq\bp$ is a necklace, and
    \item  $\beta_1$ and $\beta_2$ be the necklaces defined in Remark~\ref{rem:granddaddy} for $\bs$.
\end{itemize}
Recall that $\beta_1$ has suffix $\bp$ and $\beta_2$ has prefix $\bq$.

To adapt {\bf Step 1} from the previous section, observe that necklaces $\beta_1$ and $\beta_2$ may not satisfy the required weight constraint. This is illustrated in the following example.

%------------------------
\begin{exam}  
Consider $n=3$, $k=4$, and $w=9$.  Since 
\[  \bN_4(3,9^\uparrow)  = \tt{144,~ 234,~ 243,~ 244,~ 333,~ 334,~ 344,~ 444}, \]
we have
\[ \mathcal{G}_4(3,9^\uparrow) = 
\tt{14\underline{4\cdot 23}4\cdot
243\cdot
244\cdot
3\cdot
334\cdot
344\cdot
4}.\]
Consider $\bs = \tt{423}$ where $\bp = \tt{4}$ and $\bq = \tt{23}$. In $\mathcal{G}_4(3)$, $\bs$ is found as a substring of $\beta_1 \cdot \beta_2 = \tt{22\underline{4 \cdot 23}3}$. However, in $\mathcal{G}_4(3,9^\uparrow)$ it is found as a substring of $14\underline{4 \cdot 23}4$. %Similarly,  $\tt{414}$ is found as a substring of $\tt{134 \cdot 142}$ in $\mathcal{G}_4(3)$, but is found in the wraparound of $\mathcal{G}_4(3,9)$.    
\end{exam}
%------------------------
%

\noindent
To address this challenge, first consider the strings found in the wraparound of $\mathcal{G}_k(n,w^\uparrow)$.  
Recall from Property~\ref{prop:first} that $\alpha_1$ is aperiodic. Thus, $\alpha_1$ is a prefix of $\mathcal{G}_k(n,w^\uparrow)$. 
    %it is aperiodic and  the form $1^j\tt{x}k^{n-j-1}$ for some $\tt{x} \in \Sigma_k$ and $j<n$; it is clearly aperiodic. 
    From Property~\ref{prop:last}, $\alpha_{m-1} = (\tt{k{-}1})\tt{k}^{n-1}$ and $\alpha_m = \tt{k}^n$. Thus, $\mathcal{G}_k(n,w^\uparrow)$ has suffix $\tt{k}^n$.  
Therefore, if $\bs$ is of the form $\tt{k}^{n-j}\tt{a}_1\cdots \tt{a}_j$ for some $0 \leq j < n$, then it occurs as one of the last $n$ substrings of $\mathcal{G}_k(n,w^\uparrow)$. The rank of such a string $\bs$ is $|\bS_k(n,w^\uparrow)|-(n-j)+1$.  
For all other possible strings $\bs$ consider necklaces $\delta_1$, $\delta_2$, and $\delta_3$ such that

\begin{itemize}
 \item $\delta_1$ is the \emph{largest} necklace in $\bN_k(n,w^\uparrow)$ that is lexicographically smaller than or equal to $\beta_1$,  
 \item $\delta_2$ is the \emph{smallest} necklace in $\bN_k(n,w^\uparrow)$ that is lexicographically larger than or equal to $\beta_2$, and
 \item $\delta_3$ is the necklace following $\delta_2$ in $\bN_k(n,w^\uparrow)$.
\end{itemize}
By their definitions $\delta_1$, $\delta_2$, and $\delta_3$ appear consecutively in $\bN_k(n,w^\uparrow)$.  The necklaces $\delta_2$ and $\delta_3$ are well defined due to the constraints on $\bs$ and $w$.  The fact that $\delta_1$ is well defined follows from the proof of the following lemma. Ultimately, we will show that $\bs = \bp\bq$ is found in $\mathcal{G}_k(n,w^\uparrow)$ as a substring of $\ap(\delta_1)\ap(\delta_2)\ap(\delta_3)$.

%=============================================
\begin{lemma} \label{lem:d1}
    The string $\ap(\delta_1)$ has suffix $\bp$.
\end{lemma}
\begin{proof}
First, we show that $\delta_1$ has suffix $\bp$. Recall that necklace $\beta_1$ has suffix $\bp$. Let $j = |\bp|$ and let $\beta_1 = \tt{b}_1\cdots \tt{b}_{n-j}\bp$. If $\bp \neq \tt{k}^j$, then by Property~\ref{prop:lastnon} it must be that $\beta_2$ also has prefix $\tt{b}_1\cdots \tt{b}_{n-j}$ which  is equal to $\bq$ by the definition of $\beta_2$. Thus $\beta_1 = \bq\bp\in\bN_k(n,w^\uparrow)$. Then by the definition of $\delta_1$, we have $\delta_1 = \beta_1$, meaning $\delta_1$ has suffix $\bp$. Now suppose $\bp = \tt{k}^j$.
 If $\beta_1\in \bN_k(n,w^\uparrow)$, then $\delta_1 = \beta_1$ and thus  $\delta_1$ has suffix $\bp$. Otherwise $\delta_1<\beta_1$ and $\beta_1\not\in \bN_k(n,w^\uparrow)$. Suppose, for a contradiction, that $\delta_1$ does not have $\bp$ as a suffix. If the length-$(n-j)$ prefix of $\delta_1$ is $\tt{b}_1\cdots \tt{b}_{n-j}$, then the weight of $\beta_1$ is greater than the weight of $\delta_1$, a contradiction. If the length-$(n-j)$ prefix of $\delta_1$ is not $\tt{b}_1\cdots \tt{b}_{n-j}$, then it must be smaller than $\tt{b}_1\cdots \tt{b}_{n-j}$. But then by Property~\ref{prop:lastnon}, there is a necklace in $\bN_k(n, w^\uparrow)$ that is larger than $\delta_1$ but smaller than $\beta_1$, a contradiction. Thus $\delta_1$ has suffix $\bp$.

Now we show that $\ap(\delta_1)$ has suffix $\bp$. If $\delta_1$ is aperiodic, then we have already proved that $\bp$ is a suffix of $\ap(\delta_1)=\delta_1$. If $\delta_1$ is periodic then it can be written as $\delta_1 = \gamma^j$ for some $j > 1$ where $\ap(\delta_1) = \gamma$.  If $|\bp| \leq |\gamma|$ then $\bp$ is a suffix of $\ap(\delta_1)$. Otherwise, we have that $\ap(\delta_1)\ap(\delta_2)$ has prefix $\delta_1$ by Property~\ref{prop:periodic}. This implies that $\ap(\delta_2)$ begins with $\gamma^{j-1}$. Since $|\gamma| < |\bp|$, we have that $|\bq| = |\delta_1|-|\bp| < |\delta_1| - |\gamma| = |\gamma^{j-1}|$, which means that $\bq$ is a prefix of $|\gamma^{j-1}|$, and is therefore a prefix of $\delta_1$. But then $\delta_1=\bq\bp$ and $\delta_1$ is a cyclic shift of $\bs$, which means that $\bs$ is periodic with the same period as $\delta_1$. This contradicts $\bq$ being the longest suffix of $\bs$ such that $\bq\bp$ is a necklace.  Thus $\bp$ is a suffix of $\ap(\delta_1)$.
\end{proof}
%=============================================

\begin{lemma} \label{lem:d2}
   The string $\ap(\delta_2)\ap(\delta_3)$ has prefix $\bq$.
\end{lemma}
\begin{proof}
First, we show that $\delta_2$ has prefix $\bq$. Recall that $\beta_2$ has prefix $\bq$.  Let $j = |\bp|$.  Since $\bs = \bp\bq$ is in $\bS_k(n,w^\uparrow)$ then $\bq\tt{k}^j$ is also in $\bS_k(n,w^\uparrow)$.
Since $\bq\bp$ is a necklace, then  $\bq\tt{k}^j$ is a necklace by repeated application of Property~\ref{prop:lastnon}.  Thus, since $\delta_2$ is the smallest necklace in $\bN_k(n,w^\uparrow)$ that is lexicographically larger than $\beta_2$, $\delta_2$ must have prefix $\bq$.

Now we show that $\ap(\delta_2)\ap(\delta_3)$ has prefix $\bq$. If $\delta_2$ is aperiodic then clearly $\bq$ is a prefix of $\ap(\delta_2)$. 
If $\delta_2$ is periodic, then it cannot be $\tt{k}^n$ by the restriction on $\bs$. Additionally, we have that $\ap(\delta_2)\ap(\delta_3)$ has prefix $\delta_2$ by Property~\ref{prop:periodic}. Thus $\bq$ is a prefix of $\ap(\delta_2)\ap(\delta_3)$.  
\end{proof}

%=============================================

\begin{corollary}
The string $\bs = \bp\bq$ is a substring of $\ap(\delta_1)\ap(\delta_2)\ap(\delta_3)$.
\end{corollary}

%=============================================

To adapt {\bf Step 2} from the previous section, we compute 
$T_k(n,w,\alpha)$ which counts the number of strings in $\bS_k(n,w^\uparrow)$ whose necklaces are lexicographically smaller than a necklace $\alpha$.
Note that $T_k(n,w,\alpha_i) = |\ap(\alpha_1)\ap(\alpha_2)\cdots\ap(\alpha_{i-1})|$. 
Let $\rank'(\bs)$ denote the rank of $\bs$ in $\mathcal{G}_k(n,w^\uparrow)$
and let $S_k(n,w^\uparrow) = |\bS_k(n,w^\uparrow)|$.
Putting this all together, we can compute $\rank'(\bs)$ as follows.

\begin{result} \noindent

\vspace{-0.05in}

\begin{center}
$\rank'(\bs) = \left\{ \begin{array}{ll}
         S_k(n,w^\uparrow)-(n-j)+1    \    &\ \  \mbox{if $\bs = \tt{k}^{n-j}\tt{a}_1\cdots \tt{a}_j$ for some $0 \leq j < n$,}\\
         T_k(n, w,\beta_2) - |\bp|+1 \  &\ \  \mbox{otherwise.} 
         \end{array} \right. $
\end{center}

\vspace{-0.1in}
\end{result}
\noindent
Note that is not necessary to compute $\delta_2$, since $T_k(n,w,\delta_2) = T_k(n,w,\beta_2)$.

It remains to analyze the running time required to compute $\rank'(\bs)$.  The following recurrence can be applied to enumerate $S_k(n,w^\uparrow)$ for $n,k\geq 0$ and  $0 \leq w \leq kn$:
\begin{equation}
    S_k(n,w^\uparrow) = \sum_{j=1}^{k} S_k(n-1,\max(0,w-j)^\uparrow), 
\end{equation}
with base cases $S_k(n,w^\uparrow) = k^n$ for $w \leq n$ and $S_k(0,w^\uparrow) = 0$ for $w > 0$.
Applying dynamic programming, the value $S_k(n,w^\uparrow)$ can be computed in $\mathcal{O}(n^2k^2)$ time and   $\mathcal{O}(n^2k)$ space.   In Section~\ref{sec:rank}, the following lemma is proved.
%
%-------------------------------
\begin{lemma} \label{lem:T}
Let $n,k > 1$. For any $\alpha \in \bN_k(n)$ the value $T_k(n,w,\alpha)$ can be computed in $\mathcal{O}(n^3k^2)$ time using $\mathcal{O}(n^3k)$ space. 
\end{lemma}

\noindent
Recall that $\beta_2$ can be computed in $\mathcal{O}(n^2)$ time using $\mathcal{O}(n)$ space.  Putting this all together we obtain the following theorem.

\noindent

%-------------------------------
\begin{theorem}
The universal cycle $\mathcal{G}_k(n,w^\uparrow)$  can be ranked in $\mathcal{O}(n^3k^2)$ time using $\mathcal{O}(n^3k)$ space.
\end{theorem}
%-------------------------------

\subsection{Unranking}

%Let $1 \leq r \leq |\mathcal{G}_k(n,w^\uparrow)|$.  

In this section, we demonstrate how to determine the string $\bs = \tt{s}_1\tt{s}_2\cdots \tt{s}_n$ at rank $r$ in $\mathcal{G}_k(n,w^\uparrow)$.   We consider two cases depending on the value of $r$.

{\bf Case 1:}  $S_k(n,w^\uparrow) - n +1 < r \leq S_k(n,w^\uparrow)$.  In this case, $\bs$ is found in the wraparound.  Recall that $\tt{a}_1\cdots \tt{a}_n$ are the first $n$ symbols in $\mathcal{G}_k(n,w^\uparrow)$ as defined in Property~\ref{prop:first}, and the last $n$ symbols  are $\tt{k}^n$ (from Property~\ref{prop:last}).  Thus, $\bs$ can be returned in $\mathcal{O}(n)$ time.

{\bf Case 2 :} $1 \leq r  \leq S_k(n,w^\uparrow) - n +1$.
By the construction of $\mathcal{G}_k(n,w^\uparrow)$ and Property~\ref{prop:periodic}, the ranks of the necklaces as they appear in $\bN_k(n,w^\uparrow)$ are strictly increasing.  Let {\sc SmallestNeck}($r$) denote the smallest necklace in $\bN_k(n,w^\uparrow)$ with rank  greater than or equal to $r$ in $\mathcal{G}_k(n,w^\uparrow)$.  If $\gamma_1 =$ {\sc SmallestNeck}($r$), then let $r'$ denote the rank of $\gamma_1$ in $\mathcal{G}_k(n,w^\uparrow)$.  If $r = r'$, then $\bs = \gamma_1$.  Otherwise, let $\gamma_2 = $
{\sc SmallestNeck}($r'-n$). By Property~\ref{prop:periodic}, $\gamma_2$ and $\gamma_1$ appear consecutively in $\bN_k(n,w^\uparrow)$, at most one of the two necklaces is periodic, and $\bs$ is obtained by concatenating the last $r'-r$ elements from $\gamma_2$  with the first $n-(r'-r)$ elements from $\gamma_1$. 
\begin{exam}
Let $n=3, k=4$, and $w=9$, and consider $r=3$. Recall 
that  \[ \mathcal{G}_4(3,9^\uparrow) = 
\tt{14\underline{4\cdot 23}4\cdot
243\cdot
244\cdot
3\cdot
334\cdot
344\cdot
4}.\]
 The necklace $\gamma_1 =$ {\sc SmallestNeck}($r$) = $\tt{234}$ has rank $r' = 4$ and
$\gamma_2 =$ {\sc SmallestNeck}$(r'-n) = \tt{144}$. Since $r'-r=1$, the string at rank $r=3$ is the last symbol of $\gamma_2$ concatenated with the first $n-1=2$ symbols of $\gamma_1$. It is  $\bs=423$.

\end{exam}

\noindent  
The necklace {\sc SmallestNeck}($r$) can be computed by making $\mathcal{O}(n \log k)$ rank queries as illustrated in Algorithm~\ref{algo:unrank}.  It follows the same approach as the algorithm in~\cite{ranking} adding the weight constraint.  The algorithm iteratively determines each symbol of $\tau = \bt_1\cdots \bt_n$ starting with the first symbol. The loop invariant after $i$ iterations maintains the property that $\bt_1\cdots \bt_i\tt{k}^{n-i}$ is a necklace with rank greater than or equal to $r$ such that $\bt_1\cdots \bt_{i-1}(\bt_{i}-1)\tt{k^{n-i}}$ is either not in $\bN_k(n,w^\uparrow)$, or it is in $\bN_k(n,w^\uparrow)$ and has rank less than $r$. Binary search, is applied to determine the $i$-th symbol using at most $\log k$ ranking queries.  

%==================================================
\begin{algorithm}[htb]
\small
\caption{Compute the smallest necklace $\bt_1\bt_2\cdots \bt_n$ in $\bN_k(n,w^\uparrow)$ with rank greater than or equal to $r$, where $1 \leq r \leq S_k(n,w^\uparrow) - r + 1$}\label{algo:unrank}
\begin{algorithmic}[1]
\Statex
\Function{SmallestNeck}{$r$}
\Statex

	\State \blue{$\triangleright$ Apply a binary search to find each $\bt_i$}
        \For {$i$ {\bf from} $1$ {\bf to} $n$} 
        		\State $min \gets 1$
		        \State $max \gets k$
                \State $\bt_i \gets \tt{k}$
        		\While{$min < max$}	
			\State $prev \gets \bt_i$
			 \State $v \gets (min + max)/2$
			 \State $\bt_i \gets v$
			 \If{ $\bt_1\cdots \bt_i\tt{k}^{n-i} \in \bN_k(n,w^\uparrow)$ {\bf and} $\rank'(\bt_1\cdots \bt_i\tt{k}^{n-i}) \geq r$}     \ \ $max \gets v$
			 \Else
			 	\State $\bt_i \gets prev$
				\State $min \gets v+1$
			\EndIf
		\EndWhile
	\EndFor	

	\State \Return $\bt_1\bt_2\cdots \bt_n$
	
 \EndFunction
      
\end{algorithmic}
\end{algorithm}

%-------------------------------
\begin{theorem}
The universal cycle $\mathcal{G}_k(n,w^\uparrow)$  can be unranked in $\mathcal{O}(n^4k^2\log k)$ time using $\mathcal{O}(n^3k)$ space.
\end{theorem}
%-------------------------------
\begin{proof}
Testing whether a string is a necklace can be determined in $\mathcal{O}(n)$ time/space~\cite{Booth}.  Thus, the running time and space for the described unranking algorithm is dominated by the $\mathcal{O}(n \log k)$ calls to the ranking function.  
\end{proof}

%=====================================================================
%=====================================================================
%=====================================================================
\section{An upper bound on weight} \label{sec:upper}
Recall that  $\bS_k(n,w_\downarrow)$ denotes the subset of $\bS_k(n)$ containing strings with weight \emph{at most} $w$.
 %In this section, we demonstrate an efficiently decodable universal cycle for $\bS'_k(n,w)$.
Let the \defo{complement} of a string $\bs$, denoted by $\comp(\bs)$, be obtained by interchanging each symbol $\tt{x}$ with $\tt{k{-}x{+}1}$.  
 If $\bs$ has weight $w$, then $\comp(\bs)$ has weight $kn-w+n$.
%
%-------------------------------
\begin{exam}
Let $n=5$, $k=4$, and consider $\bs = \tt{12443}$. It has weight $w=14$.  Then $\comp(\bs) = \tt{43112}$ has weight $20-14+5 = 11$; the symbols $\tt{1}$ and $\tt{4}$ are interchanged, and the symbols $\tt{2}$ and $\tt{3}$ are interchanged.  
\end{exam}
%-------------------------------
%

\noindent
The complement of a universal cycle for $\bS_k(n,(kn-w+n)^\uparrow)$ is a universal cycle for $\bS(n,w_\downarrow)$, and vice versa.  Thus, $\mathcal{G}_k(n,w_\downarrow) = \comp(\mathcal{G}_k(n,kn-w+n)^\uparrow)$ is a universal cycle for $\bS_k(n,w_\downarrow)$.  To rank a string $\bs$ in $\mathcal{G}_k(n,w_\downarrow)$ simply return the rank of $\comp(\bs)$ in $\mathcal{G}_k(n,(kn-w+n)^\uparrow)$.

%-------------------------------
\begin{corollary} \label{cor:S2}
The universal cycle $\mathcal{G}_k(n,w_\downarrow)$  can be ranked in $\mathcal{O}(n^3k^2)$ time and unranked in $\mathcal{O}(n^4k^2\log k)$ time  using $\mathcal{O}(n^3k)$ space.
\end{corollary}
%-------------------------------

\begin{comment}
From~\cite{weight1,concattree}, $\mathcal{G}_k(n,w^\uparrow)$ can be constructed in $\mathcal{O}(1)$-amortized time per symbol using $\mathcal{O}(n^2)$ space.  Thus, the same analysis can be applied to $\mathcal{G}_k(n,w^\uparrow)$ by simply complementing the output of each symbol.

%-------------------------------
\begin{theorem} \label{thm:construct}
Let $n,k > 2$ and $k \leq w \leq nk$.  Then $\mathcal{G}_k(n,w^\uparrow)$ and $\mathcal{W}_k^{\downarrow}(n,w)$ can be constructed in $\mathcal{O}(1)$-amortized time using $\mathcal{O}(n^2)$ space.
\end{theorem}
%-------------------------------
\end{comment}

%=====================================================================
%=====================================================================
%=====================================================================
\section{Applications to $t$-subsets and $t$-multisets} \label{sec:subset}

Recall that $\mathbf{Subset}(n,t)$ denotes the set of all $t$-subsets of $\{\tt{1,2,\ldots , n}\}$. 
Using difference representatives for each subset, a universal cycle for $\mathbf{Subset}(n,t)$ is precisely  a universal cycle for  $\mathbf{S}_{n-t+1}(t,n_\downarrow)$~\cite{difference}.  
Thus, we can apply the results from Corollary~\ref{cor:S2} to efficiently rank a universal cycle for $\mathbf{Subset}(n,t)$.  
%

%-------------------------------
\begin{theorem}
 The sequence $\mathcal{G}_{n-t+1}(t,n_\downarrow)$ is a universal cycle for $\mathbf{Subset}(n,t)$ applying the difference representatives. 
 %It can be generated in $\mathcal{O}(1)$-amortized time per symbol using $O(t^2)$ space.  
 Using $\mathcal{O}(nt^3)$ space, the sequence can be ranked in $\mathcal{O}(n^2t^3)$ time and unranked in $\mathcal{O}(n^2t^4\log n)$ time.
\end{theorem}
%-------------------------------

\begin{exam}
%Consider $\mathbf{Subset}(5,3)$.  
The sequence $\mathcal{G}_3(3,5_\downarrow) = \tt{3112212111}$ is a universal cycle for $\mathbf{Subset}(5,3)$ using difference representatives.  It can be efficiently decoded using its complement
$\mathcal{G}_3(3,7^\uparrow) = \tt{1332232333}$ as illustrated in the following table.
\begin{center}
\begin{tabular}{c|c|c|c}
  {\bf Rank}   & {\bf String in $\mathcal{G}_3(3,7^\uparrow)$}  & {\bf Complement (diff rep)} & {\bf Subset over $\{1,2,3,4,5\}$}  \\ \hline
    1   & $\tt{133}$   & $\tt{311}$   & $\{\tt{3,4,5}\}$ \\   
    2   & $\tt{332}$   & $\tt{112}$   & $\{\tt{1,2,4}\}$ \\   
    3   & $\tt{322}$   & $\tt{122}$   & $\{\tt{1,3,5}\}$ \\   
    4   & $\tt{223}$  & $\tt{221}$   & $\{\tt{2,4,5}\}$ \\   
    5   & $\tt{232}$   & $\tt{212}$   & $\{\tt{2,3,5}\}$ \\   
    6   & $\tt{323}$   & $\tt{121}$   & $\{\tt{1,3,4}\}$ \\   
    7   & $\tt{233}$   & $\tt{211}$   & $\{\tt{2,3,4}\}$ \\   
    8   & $\tt{333}$   & $\tt{111}$   & $\{\tt{1,2,3}\}$ \\   
    9   & $\tt{331}$   & $\tt{113}$   & $\{\tt{1,2,5}\}$ \\   
    10  & $\tt{313}$   & $\tt{131}$   & $\{\tt{1,4,5}\}$ \\   
\end{tabular}
\end{center}

\vspace{-0.15in}
 
\end{exam}

Recall that $\mathbf{Multiset}(n,t)$ denotes the set of all $t$-multisets of $\{\tt{0,2,\ldots , n{-}1}\}$. 
Using difference representatives for each multiset and incrementing all symbols by 1, a universal cycle for $\mathbf{Multiset}(n,t)$ is precisely  a universal cycle for  $\mathbf{S}_{n}(t,(n+t-1)_\downarrow)$~\cite{difference}.  
Again, we can apply the results from Corollary~\ref{cor:S2} to efficiently decode a universal cycle for $\mathbf{Multiset}(n,t)$.  
%

%-------------------------------
\begin{theorem}
 The sequence $\mathcal{G}_{n}(t,(n+t-1)_\downarrow)$ is a universal cycle for $\mathbf{Multiset}(n,t)$ applying the adjusted difference representatives. 
 %It can be generated in $\mathcal{O}(1)$-amortized time per symbol using $O(t^2)$ space.  
 Using $\mathcal{O}(nt^3)$ space, the sequence can be ranked in $\mathcal{O}(n^2t^3)$ time and unranked in $\mathcal{O}(n^2t^4\log n)$ time.
\end{theorem}
%-------------------------------

\begin{exam}
%Consider $\mathbf{Subset}(5,3)$.  
The sequence $\tt{2001101000}$ is a universal cycle for $\mathbf{Multieset}(3,3)$ using difference representatives. By incrementing each symbol in the sequence by one, we obtain $\mathcal{G}_3(3,5_\downarrow) = \tt{3112212111}$.
It can be efficiently decoded using its complement
$\mathcal{G}_3(3,7^\uparrow) = \tt{1332232333}$ as illustrated in the following table.
\begin{center}
\begin{tabular}{c|c|c|c}
  {\bf Rank}   & {\bf String in $\mathcal{G}_3(3,7^\uparrow)$}  & {\bf Complement minus 1 (diff rep)} & {\bf Multiset over $\{0,1,2\}$}  \\ \hline
    1   & $\tt{133}$   & $\tt{200}$   & $\{\tt{2,2,2}\}$ \\   
    2   & $\tt{332}$   & $\tt{001}$   & $\{\tt{0,0,1}\}$ \\   
    3   & $\tt{322}$   & $\tt{011}$   & $\{\tt{0,1,2}\}$ \\   
    4   & $\tt{223}$  & $\tt{110}$   & $\{\tt{1,2,2}\}$ \\   
    5   & $\tt{232}$   & $\tt{101}$   & $\{\tt{1,1,2}\}$ \\   
    6   & $\tt{323}$   & $\tt{010}$   & $\{\tt{0,1,1}\}$ \\   
    7   & $\tt{233}$   & $\tt{100}$   & $\{\tt{1,1,1}\}$ \\   
    8   & $\tt{333}$   & $\tt{000}$   & $\{\tt{0,0,0}\}$ \\   
    9   & $\tt{331}$   & $\tt{002}$   & $\{\tt{0,0,2}\}$ \\   
    10  & $\tt{313}$   & $\tt{020}$   & $\{\tt{0,2,2}\}$ \\   
\end{tabular}
\end{center}

\vspace{-0.15in}
 
\end{exam}

%=====================================================================
%=====================================================================
%=====================================================================

%=====================================================================
%\section{Ranking necklaces with bounded weight} \label{sec:rank}
\section{Proof of Lemma~\ref{lem:T}: Enumerating $T_k(n,w,\alpha)$} \label{sec:rank}

In this section we demonstrate how to efficiently compute $T_k(n,w,\alpha)$, where $\alpha$ is a necklace in $\bN_k(n,w^\uparrow)$. 
%We also show how this value can be used to efficiently rank and unrank the necklaces in $\bN_k(n,w^\uparrow)$.  
The approach follows the methods as applied in~\cite{ranking,hartman}.
To simplify notation going forward assume that $n,k$ and the necklace $\alpha = \tt{a}_1\tt{a}_2\cdots \tt{a}_n$ are fixed.

%=====================================================================
\subsection{Enumerating preliminary sets} \label{sec:B}

In order to efficiently compute compute $T_k(n,w,\alpha)$, we need to efficiently determine the cardinality of two special sets of strings that have both prefix and suffix restrictions.  For $t \geq j$, let  $\mathbf{B}(t,j,w)$ denote the subset of strings in $\bS_k(t,w)$  where each string has prefix $\tt{a}_1\tt{a}_2\cdots \tt{a}_{j}$ and every non-empty suffix is greater than $\alpha$. 
  When $j=0$, there is no prefix restriction on the strings.

\begin{exam}
Let $\alpha = \tt{112122}$ and $k=3$.  Then the strings in $\mathbf{B}(6,3,11)$ are:

\begin{center}
\begin{tabular}{c c c}
\textcolor{blue}{$\tt{112}$}\textcolor{red}{$\tt{1}$}$\tt{33}$  & \textcolor{blue}{$\tt{112}$}\textcolor{red}{$\tt{2}$}$\underline{\tt{23}}$ & \textcolor{blue}{$\tt{112}$}\textcolor{red}{$\tt{3}$}$\underline{\tt{13}}$ \\
                            & \textcolor{blue}{$\tt{112}$}\textcolor{red}{$\tt{2}$}$\underline{\tt{32}}$ & \textcolor{blue}{$\tt{112}$}\textcolor{red}{$\tt{3}$}$\underline{\tt{22}}$ \\
                            & \textcolor{blue}{$\tt{112}$}\textcolor{red}{$\tt{2}$}$\underline{\tt{33}}$ & \textcolor{blue}{$\tt{112}$}\textcolor{red}{$\tt{3}$}$\underline{\tt{23}}$ \\
                             &  & \textcolor{blue}{$\tt{112}$}\textcolor{red}{$\tt{3}$}$\underline{\tt{32}}$ \\
                             &  & \textcolor{blue}{$\tt{112}$}\textcolor{red}{$\tt{3}$}$\underline{\tt{33}}$ \\
\end{tabular}
\end{center}
%Note that $B(6,3,11) = 9$. 
The string in the first column can be given by $\mathbf{B}(6,4,11)$,
while the  underlined strings in the second and third columns correspond to the sets $\mathbf{B}(2,0,5)$ and $\mathbf{B}(2,0,4)$, respectively.
\end{exam}

\begin{comment}
%-----------------------------
\begin{exam}
Let $\alpha=\T{aaabcc}$ and $\Sigma = \{\T{a,b,c}\}$.   $\mathbf{B}_\alpha(2,0) = \{\T{ab,ac,bb,bc,cb,cc}\} $ consists of strings of length 2 where every non-empty suffix is greater than $\alpha$.  The strings in $\{\T{aa,ba,ca}\}$ are not in this set because the suffix \T{a} is lexicographically smaller than $\alpha$.   Now consider $\mathbf{B}_\alpha(5,2)$ partitioned on the 3rd  (underlined) character: 
%
\vspace{-0.25in}
\begin{center}
\begin{tabular}  {l  @{\hskip 0.4in}  l @{\hskip 0.4in}  l }  \\
\T{aa\blue{\underline{a}}\:cb}  &   \T{aa\blue{\underline{b}}\:ab}   & \T{aa\blue{\underline{c}}\:ab}   \\
\T{aa\blue{\underline{a}}\:cc}  &   \T{aa\blue{\underline{b}}\:ac}   & \T{aa\blue{\underline{c}}\:ac}   \\
                            &   \T{aa\blue{\underline{b}}\:bb}   & \T{aa\blue{\underline{c}}\:bb}   \\
                             &  \T{aa\blue{\underline{b}}\:bc}   & \T{aa\blue{\underline{c}}\:bc}   \\
                             &   \T{aa\blue{\underline{b}}\:cb}   & \T{aa\blue{\underline{c}}\:cb}   \\
                             &   \T{aa\blue{\underline{b}}\:cc}   & \T{aa\blue{\underline{c}}\:cc}.   \\
\end{tabular}
\end{center}

\noindent
Observe how this partition decomposes the set revealing recursive structure:  
\[\mathbf{B}_\alpha(5,2) = \mathbf{B}_\alpha(5,3) \  \cup \  \T{aab} \cdot \mathbf{B}_\alpha(2,0)  \ \cup  \  \T{aac} \cdot \mathbf{B}_\alpha(2,0).\]
%
\end{exam}
%-----------------------------
\end{comment}

%
\noindent
Let $B(t,j,w) = |\mathbf{B}(t,j,w)|$.
An enumeration formula for $B(t,j,w)$, where $0\leq j\leq t \leq n$, can be derived based on the recursive structure illustrated in the previous example.  Considering the empty string, $B(0,0,w) = 1$ if $w \leq 0$ and $B(0,0,w) = 0$ if $w > 0$.  When $t > 0$,  $B(t,t,w) = 0$ because  the suffix $\tt{a}_1\tt{a}_2\cdots \tt{a}_t$ is  lexicographically smaller than or equal to~$\alpha$.  Otherwise, the strings in $\mathbf{B}(t,j,w)$ can be partitioned based on the symbol in position $j{+}1$.
\squishlisttwo
\item If the $j{+}1$st symbol is smaller than $\tt{a}_{j+1}$  then the suffix starting from the first index would be smaller than $\alpha$, a contradiction.  
\item If the $j{+}1$st  symbol is $\tt{a}_{j+1}$ then the number of such strings is $B(t,j+1,w)$.  
\item If the $j{+}1$st  symbol is larger than $\tt{a}_{j+1}$, then any suffix starting at index $1,2,\ldots, j+1$ is 
larger than $\alpha$ by Lemma~\ref{prop:suffix}.  Thus, for each $\tt{x} > \tt{a}_{j+1}$, the remaining $t-j-1$ positions can be filled in $B(t-j-1,0,w-\wt(\tt{a}_1\cdots \tt{a}_t\tt{x}))$ ways.  
\squishend
Thus, for $0\leq j\leq t \leq n$ and $w \geq 0$:
\begin{center} \small
$B(t,j,w)  = \left\{ \begin{array}{ll}
    1	 \ & \ \  \   \mbox{if \ \ $j=t=0$ and $w \leq 0$, } \\
    0	 \ & \ \  \  \mbox{if \ \ $j=t$,} \\
    B(t,j+1,w) + \displaystyle{\sum_{\tt{x} = \tt{a}_{j+1}+1}^\tt{k} B(t-j-1,0, w-\wt(\tt{a}_1\cdots \tt{a}_j\tt{x})) } \  &\ \ \  \mbox{otherwise.}
    \end{array} \right.
            \label{eq:B}$
\end{center}
%

%================
Now consider the set $\mathbf{P}(t,j,w)$ which denotes the subset of  $\bS_k(t)$ where each string has weight \emph{exactly} $w$ and prefix  $\tt{a}_1\tt{a}_2\cdots \tt{a}_{j}$ such that every non-empty suffix is greater than $\alpha$. Let $P(t,j,w) = |\mathbf{P}(t,j,w)|$.  Applying the same recursive strategy applied for computing $B(t,j,w)$, but with different base cases,  we obtain the following recurrence for $0\leq j \leq t$:
\[
\small 
{P}(t,j,w) = 
\begin{cases}	
	1 & \ \ \  \text{if } j=t=w=0, \\
	0 & \ \ \  \text{if } j=t \mbox{ or } w \leq 0,  \\
	{P}(t,j+1,w) + \displaystyle{
        \sum_{\tt{x}=\tt{a}_{j+1}+1}^{\tt{k}}P(t-j-1,0,w -\wt(\tt{a}_1\cdots \tt{a}_j\tt{x}))} & \ \ \ \text{otherwise}. \\
\end{cases}
\]
\normalsize

%=====================================================================
\subsection{Computing $T_k(n,w,\alpha)$}

Let $\textbf{T}_k(n,w,\alpha)$ denote the subset of strings in $\bS_k(n,w^\uparrow)$ whose necklaces are lexicographically smaller than a given necklace $\alpha$.  Note that $T_k(n,w,\alpha) = |\textbf{T}_k(n,w,\alpha)|$.  To compute $T_k(n,w,\alpha)$, we partition $\textbf{T}_k(n,w,\alpha)$ into subsets following the approaches used in~\cite{ranking,hartman}.  

Let $\textbf{A}(t,j)$ denote the subset of strings  in $\textbf{T}_k(n,w,\alpha)$ of the form $\bs= \tt{s}_1\tt{s}_2\cdots \tt{s}_{n}$ where $t$ is the smallest integer such that $\bs'= \tt{s}_t\tt{s}_{t+1}\cdots \tt{s}_1\tt{s}_2\cdots \tt{s}_{t-1} < \alpha$ and $j$ is the largest integer such that $\tt{a}_1 \cdots \tt{a}_j$ is a prefix of $\bs'$.   Let $A(t,j) =  |\textbf{A}(t,j)|$. %
 Then
\[
T_k(n,w, \alpha) =  \displaystyle{\sum_{t=1}^{n}\sum_{j=0}^{n} A(t,j).}
\]

%Consider for some set $\textbf{A}_k(t,j,w, \alpha)$, what would the $j+1$'st element be? Well we know that the string is smaller then or equal to $\alpha$ and that it can't be equal to $\tt{a}_{j+1}$ as that would imply the string is apart $\textbf{A}_k(t,j+1,w, \alpha)$. Thus we can see that the $j+1'st$ element must be smaller than $\tt{a}_{j+1}$.
%EXAMPLE 3

\begin{exam}
\begin{comment}
 The sequence of necklaces $\bN_3(5,10)$ is \small
    \[ 
\tt{1 1 2 3 3,
1 1 3 2 3,
1 1 3 3 2,
1 1 3 3 3,
1 2 1 3 3,
1 2 2 2 3,
1 2 2 3 2,
1 2 2 3 3,
\underline{1 2 3 1 3},
1 2 3 2 3,
1 2 3 3 2,
1 2 3 3 3,
1 3 1 3 2,
1 3 1 3 3, 
1 3 2 2 2,} \] 
%
\vspace{-0.25in}
%
\[
\tt{1 3 2 2 3,
1 3 2 3 2,
1 3 2 3 3,
1 3 3 2 2,
1 3 3 2 3,
1 3 3 3 2,
1 3 3 3 3,
2 2 2 2 2,
2 2 2 2 3,
2 2 2 3 3,
2 2 3 2 3,
2 2 3 3 3,
2 3 2 3 3,
2 3 3 3 3,
3 3 3 3 3. }
\]
\end{comment}
\noindent
The set $\textbf{T}_3(5,10,\tt{12313})$ contains 40 strings from the equivalence classes
represented by the eight necklaces:
$\{ \tt{1 1 2 3 3,
1 1 3 2 3,
1 1 3 3 2,
1 1 3 3 3,
1 2 1 3 3,
1 2 2 2 3,
1 2 2 3 2,
1 2 2 3 3}\}$
 that appear before $\tt{12313}$ in $\bN_3(5,10)$. These strings can be partitioned into 10 subsets of the form $\textbf{A}(t,j)$ 
as follows:

%We can calculate this by calculating the count of each of $\textbf{A}_3(t,j,w,\tt{12313})$ for all $1 \leq t \leq n$ and $0\leq j \leq n$. We can see how each rotation is separated off as such 

\begin{center}
\begin{tabular}{ c | c c c c c }
$\textbf{A}(t,j)$ & $t=1$ & $t=2$ & $t=3$ & $t=4$ & $t=5$ \\ 
\hline
           & $\textcolor{blue}{\tt{1}}\tt{1233}$, $\textcolor{blue}{\tt{1}}\tt{1323}$  &   $\tt{3}\textcolor{blue}{\tt{1}}\tt{123}$, $\tt{3}\textcolor{blue}{\tt{1}}\tt{132}$ &  $\tt{33}\textcolor{blue}{\tt{1}}\tt{12}$, $\tt{23}\textcolor{blue}{\tt{1}}\tt{13} $ & $\tt{223}\textcolor{blue}{\tt{1}}\tt{1}$, $\tt{323}\textcolor{blue}{\tt{1}}\tt{1}$ & $\tt{1223}\textcolor{blue}{\tt{1}}$, $\tt{1323}\textcolor{blue}{\tt{1}}$
           \\
$j=1$      & $\textcolor{blue}{\tt{1}}\tt{1332}$, $\textcolor{blue}{\tt{1}}\tt{1333}$                & $\tt{2}\textcolor{blue}{\tt{1}}\tt{133}$, $\tt{3}\textcolor{blue}{\tt{1}}\tt{113}$ & $\tt{33}\textcolor{blue}{\tt{1}}\tt{12}$, $\tt{33}\textcolor{blue}{\tt{1}}\tt{12}$ & $\tt{332}\textcolor{blue}{\tt{1}}\tt{1}$, $\tt{333}\textcolor{blue}{\tt{1}}\tt{1}$ & $\tt{1332}\textcolor{blue}{\tt{1}}$, $\tt{1333}\textcolor{blue}{\tt{1}}$
           \\
           \\
           & $\textcolor{blue}{\tt{12}}\tt{233}$, $\textcolor{blue}{\tt{12}}\tt{223}$ & $\tt{3}\textcolor{blue}{\tt{12}}\tt{13}$, $\tt{3}\textcolor{blue}{\tt{12}}\tt{22}$ & $\tt{33}\textcolor{blue}{\tt{12}}\tt{1}$, $\tt{33}\textcolor{blue}{\tt{12}}\tt{1}$ & $\tt{133}\textcolor{blue}{\tt{12}}$, $\tt{223}\textcolor{blue}{\tt{12}}$ & $\textcolor{blue}{\tt{2}}\tt{133}\textcolor{blue}{\tt{1}}$, $\textcolor{blue}{\tt{2}}\tt{223}\textcolor{blue}{\tt{1}}$  \\ 
 $j=2$ & $\textcolor{blue}{\tt{12}}\tt{232}$, $\textcolor{blue}{\tt{12}}\tt{232}$ & $\tt{2}\textcolor{blue}{\tt{12}}\tt{23}$, $\tt{3}\textcolor{blue}{\tt{12}}\tt{23}$ & $\tt{32}\textcolor{blue}{\tt{12}}\tt{2}$, $\tt{33}\textcolor{blue}{\tt{12}}\tt{2}$ & $\tt{232}\textcolor{blue}{\tt{12}}$, $\tt{233}\textcolor{blue}{\tt{12}}$  & $\textcolor{blue}{\tt{2}}\tt{232}\textcolor{blue}{\tt{1}}$, $\textcolor{blue}{\tt{2}}\tt{233}\textcolor{blue}{\tt{1}}$  
 \\
\end{tabular}
\end{center}

\noindent
For $j \in \{0,3,4,5\}$, $\textbf{A}(t,j) = \emptyset$ for all values of $t$. 

\end{exam}
%============================

\noindent
Clearly $A(t,n) = 0$ for all $1 \leq t \leq n$.  Otherwise for $j<n$, each set $\mathbf{A}(t,j)$ falls into one of the following two cases.
\medskip

%========
\noindent
{\bf Case 1} ($t+j \leq n$).  Each $\bs \in \mathbf{A}(t,j)$ is of the form
$\br\,\blue{\tt{a}_1\tt{a}_2\cdots \tt{a}_j}\,\tt{x}\,\bq$ where:
\squishlisttwo
\item $\br \in \bS_k(t-1,w')$ such that every non-empty suffix is larger than $\alpha$, 
\item $\tt{x} < \tt{a}_{j+1}$,
\item $\bq \in \bS_k(n-t-j,w-w'-\wt(\tt{a}_1\cdots \tt{a}_j\tt{x}))$, and 
\item $\wt(\bs) \geq w$.
\squishend

\noindent   
Recall that the number of possibilities for $\br$ with a fixed weight $w'$ is given by $P(t{-}1,0,w')$. 
For each $\br$ with weight $w'$ and given 
$\tt{x}$, there are $S_k(n-t-j,w-w'-\wt(\tt{a}_1\cdots \tt{a}_j\tt{x}))$ possibilities for $\bq$. Thus, in this case
\[ A(t,j)  = \sum_{\tt{x} = 1}^{\tt{a}_{j+1}-1}
\sum_{w'=t-1}^{k(t-1)} P(t-1,0,w')S_k(n-t-j, w-w'-\wt(\tt{a}_1\cdots \tt{a}_j\tt{x}))
. \]

%========
\noindent
{\bf Case 2} ($t+j > n$).  Each $\bs \in \mathbf{A}(t,j)$ is of the form
$\blue{\tt{a}_{n-t+2}\cdots \tt{a}_{j-1}\tt{a}_{j}}\,\tt{x}\,\br\,\blue{\tt{a}_1\tt{a}_2\cdots \tt{a}_{n-t+1}}$ where:
\squishlisttwo
\item $\tt{x} < \tt{a}_{j+1}$,
\item $\br \in \bS(n-j-1, w-\wt(\tt{a}_1\cdots \tt{a}_j\tt{x}))$, and 
\item every non-empty suffix of $\tt{a}_{n-t+2}\cdots \tt{a}_{j-1}\tt{a}_{j}\,\tt{x}\,\br$ is larger than $\alpha$.\footnote{This follows from the definition of $t$ noting that $|a_{n-t+2}\cdots a_{j-1}a_{j}\,x\,\br| < n$.}
\squishend
Let $\bq = \tt{a}_{n-t+2}\cdots \tt{a}_{j-1}\tt{a}_{j}$.  Since  $\bq$ is a substring of the necklace $\alpha$,  any suffix of $\bq$ must be lexicographically greater than or equal to the prefix of $\alpha$ with the same length (by the definition of a necklace).  
Therefore we determine the \emph{longest} suffix of $\bq$ that is equal to the prefix of $\alpha$ with the same length.  Suppose this suffix has length $z$.  This
 implies $\tt{a}_{j-z+1}\cdots \tt{a}_{j-1}\tt{a}_{j} = \tt{a}_1\tt{a}_2\cdots \tt{a}_z$.  This means any suffix of $\bs$ starting from an index less than or equal to $|
\bq|-z$ is larger than $\alpha$.  Consider three sub-cases depending on $\tt{x}$.
\squishlist
\item If $\tt{x} < \tt{a}_{z+1}$ then $\tt{a}_{j-z+1}\cdots \tt{a}_{j-1}\tt{a}_{j}\tt{x}$ is smaller than $\alpha$, a contradiction.  
\item If $\tt{x} = \tt{a}_{z+1}$ then $\tt{a}_{j-z+1}\cdots \tt{a}_{z+1}\br \in \mathbf{B}(n-j+z, z+1, w-\wt(\tt{a}_1\cdots \tt{a}_{j-z}))$.
 \item If $\tt{x} > \tt{a}_{z+1}$ then any suffix of $\bs$ starting from an index less than or equal to $|\bq|+1$ will be larger than $\alpha$ by Property~\ref{prop:suffix}.  Thus, 
$\br \in \mathbf{B}(n-j-1,0,w-\wt(\tt{a}_1\cdots \tt{a}_j\tt{x}))$.   
 %$A(t, j)= B(n-j-1,0)$.  
 \squishend
Thus, if $\tt{a}_{j+1} > \tt{a}_{z+1}$
\[ A(t,j)  = B(n-j+z, z+1, w-\wt(\tt{a}_1\cdots \tt{a}_{j-z})) + 
\sum_{\tt{x} = \tt{a}_{z+1}+1}^\tt{a_{j+1}-1} B(n-j-1,0,w-\wt(\tt{a}_1\cdots \tt{a}_j\tt{x}) ). \]
Otherwise $A(t,j) = 0$.

%======================
\subsection{Analysis}

By applying dynamic programming, all necessary $B(t,j,w)$ and $P(t,j,w)$ can be computed in $\mathcal{O}(n^3k^2)$ time using $\mathcal{O}(n^3k)$.  With these values precomputed,
%\begin{itemize}
    %\item 
    $A(t,j)$ can be computed in $O(nk^2)$ time for {\bf Case 1} and $O(k)$ time for {\bf Case 2}, and thus
 %   \item 
    $T_k(n,w,\alpha)$ can be computed in $O(n^3k^2)$ time.
%\end{itemize}
%
\noindent
This proves Lemma~\ref{lem:T}.

\section{Acknowledgment}
Joe Sawada (grant RGPIN-2025-03961) gratefully acknowledges research support from the Natural Sciences and Engineering Research Council of Canada (NSERC).

%------------------------
%\begin{exam} 
%Consider $n=3$, $k=4$, and $w=9$.  Since 
%\[  \bN^{\uparrow}_4(3,9)  = \tt{144,~ 234,~ 243,~ %244,~ 333,~ 334,~ 344,~ 444}, \]
%we have
%\[ \mathcal{W}_4^{\uparrow}(3,9) = 
%\tt{144\cdot
%234\cdot
%243\cdot
%244\cdot
%3\cdot
%334\cdot
%344\cdot
%4}.\]
%
%Consider $\alpha_1 = 312$. In $\mathcal{G}_4(3)$, $\alpha_1$ is found as a substring of $113 \cdot 122$; however, in $\mathcal{DB}_4(3,6)$ it is found as a substring of $033 \cdot 123$. Similarly,  $\alpha_2 = 303$ is found as a substring of $023 \cdot 031$ in the $\mathcal{G}_4(3)$, but is found in the wraparound of $\mathcal{DB}_4(3,6)$.    
%\end{exam}
%------------------------

%-----------  BIBLIOGRAPHY  ---------------

\bibliographystyle{alphaabbr}
\bibliography{refs}

@ARTICLE{Booth,
    author = {K. S. Booth},
    title = {Lexicographically least circular substrings},
    journal = {Inform. Process. Lett.},
    year = {1980},
    volume = {10},
    number = {4/5},
    pages = {240-242}
}

@article{hurlbert2009,
title = {On universal cycles for multisets},
journal = {Discrete Mathematics},
volume = {309},
number = {17},
pages = {5321-5327},
year = {2009},
note = {Generalisations of de Bruijn Cycles and Gray Codes/Graph Asymmetries/Hamiltonicity Problem for Vertex-Transitive (Cayley) Graphs},
issn = {0012-365X},
doi = {https://doi.org/10.1016/j.disc.2008.04.050},
url = {https://www.sciencedirect.com/science/article/pii/S0012365X08002513},
author = {Glenn Hurlbert and Tobias Johnson and Joshua Zahl}
}

@misc{difference,
  title={Universal cycles for $k$-subsets and $k$-multisets of an $n$-set},
  author={Campbell, Colin and Janik-Jones, Luke and Sawada, Joe},
  year={\emph{under review}, 2025},
  publisher={under review Discrete Mathematics}
}

@article{cantwell,
title = {Graph universal cycles of combinatorial objects},
journal = {Advances in Applied Mathematics},
volume = {127},
pages = {102166},
year = {2021},
issn = {0196-8858},
doi = {https://doi.org/10.1016/j.aam.2021.102166},
url = {https://www.sciencedirect.com/science/article/pii/S019688582100004X},
author = {Amelia Cantwell and Juliann Geraci and Anant Godbole and Cristobal Padilla}
}

@article{hartman,
title = {Ranking and unranking fixed-density necklaces and {L}yndon words},
journal = {Theoretical Computer Science},
volume = {791},
pages = {36-47},
year = {2019},
issn = {0304-3975},
doi = {https://doi.org/10.1016/j.tcs.2019.04.007},
url = {https://www.sciencedirect.com/science/article/pii/S0304397519302865},
author = {Patrick Hartman and Joe Sawada}
}

@misc{dbseq,
title={De {B}ruijn Sequence and Universal Cycle Constructions (2025). http://debruijnsequence.org.},
year = {2025}
}

@book{magic,
 author = {P. Diaconis and R. Graham},
 title = {Magical Mathematics: The Mathematical Ideas That Animate Great Magic Tricks},
 year = {2011},
 isbn = {9780691151649},
 publisher = {Princeton University Press},
}

@article{fred-nfsr,
 author = {Harold Fredricksen},
 title = {A Survey of Full Length Nonlinear Shift Register Cycle Algorithms},
 journal = {Siam Review},
 volume = {24},
 number = {2},
 year = {1982},
 pages = {195-221}
}

@book{etzion-book,
 author = {T. Etzion},
 title = {Sequences and the de {B}ruijn {G}raph},
 year = {2024},
 isbn = {9780443135170},
 publisher = {Academic Press},
}

@INPROCEEDINGS{BM,
  author={Burns, J. and Mitchell, C.J.},
  booktitle={Cryptography and Coding III (M.J.Ganley, ed.)}, 
  title={Position sensing coding schemes}, 
  year={1993},
  volume={},
  number={},
  pages={31-66},
   publisher = {Oxford University Press},
  }

@article{rudoy,
title = {An inductive approach to constructing universal cycles on the $k$-subsets of [$n$]},
journal = {The Electronic Journal of Combinatorics},
volume = {20},
pages = {P18},
year = {2013},
author = {Yevgeniy Rudoy}
}

@article{godbole2011,
  title={On universal cycles for new classes of combinatorial structures},
  author={Blanca, Antonio and Godbole, Anant P},
  journal={SIAM Journal on Discrete Mathematics},
  volume={25},
  number={4},
  pages={1832--1842},
  year={2011},
  publisher={SIAM}
}

@article {fred-lex-least,
    AUTHOR = {Fredricksen, Harold},
     TITLE = {The lexicographically least de {B}ruijn cycle},
   JOURNAL = {J. Combinatorial Theory},
    VOLUME = {9},
      YEAR = {1970},
     PAGES = {1--5},
   MRCLASS = {05.04},
  MRNUMBER = {0258641},
MRREVIEWER = {N. J. A. Sloane},
}

@article {martin,
    AUTHOR = {Martin, M. H.},
     TITLE = {A problem in arrangements},
   JOURNAL = {Bull. Amer. Math. Soc.},
  FJOURNAL = {Bulletin of the American Mathematical Society},
    VOLUME = {40},
      YEAR = {1934},
    NUMBER = {12},
     PAGES = {859--864},
      ISSN = {0002-9904},
   MRCLASS = {DML},
  MRNUMBER = {1562989},
       DOI = {10.1090/S0002-9904-1934-05988-3},
       URL = {https://doi-org.subzero.lib.uoguelph.ca/10.1090/S0002-9904-1934-05988-3},
}

@article {fredricksen-beads-colors,
    AUTHOR = {Fredricksen, Harold and Maiorana, James},
     TITLE = {Necklaces of beads in {$k$} colors and {$k$}-ary de {B}ruijn
              sequences},
   JOURNAL = {Discrete Math.},
  FJOURNAL = {Discrete Mathematics},
    VOLUME = {23},
      YEAR = {1978},
    NUMBER = {3},
     PAGES = {207--210},
      ISSN = {0012-365X},
   MRCLASS = {05A05},
  MRNUMBER = {523071},
MRREVIEWER = {N. G. de Bruijn},
       DOI = {10.1016/0012-365X(78)90002-X},
       URL = {https://doi.org/10.1016/0012-365X(78)90002-X},
}

@article {generalize-classic-greedy,
    AUTHOR = {Sawada, Joe and Williams, Aaron and Wong, Dennis},
     TITLE = {Generalizing the classic greedy and necklace constructions of
              de {B}ruijn sequences and universal cycles},
   JOURNAL = {Electron. J. Combin.},
  FJOURNAL = {Electronic Journal of Combinatorics},
    VOLUME = {23},
      YEAR = {2016},
    NUMBER = {1},
     PAGES = {Paper 1.24, 20},
      ISSN = {1077-8926},
   MRCLASS = {05A05 (05A15)},
  MRNUMBER = {3484729},
MRREVIEWER = {Sergey Kitaev},
}

@article{weight1,
title = {The lexicographically smallest universal cycle for binary strings with minimum specified weight},
journal = {Journal of Discrete Algorithms},
volume = {28},
pages = {31-40},
year = {2014},
note = {StringMasters 2012, 2013 Special Issue (Volume 1)},
issn = {1570-8667},
doi = {https://doi.org/10.1016/j.jda.2014.06.003},
url = {https://www.sciencedirect.com/science/article/pii/S1570866714000343},
author = {Joe Sawada and Aaron Williams and Dennis Wong}
}

@article{shorthand,
 author = {Ruskey, Frank and Williams, Aaron},
 title = {An Explicit Universal Cycle for the ($N$-1)-permutations of an $N$-set},
 journal = {ACM Trans. Algorithms},
 issue_date = {June 2010},
 volume = {6},
 number = {3},
 month = jul,
 year = {2010},
 issn = {1549-6325},
 pages = {1--12},
 articleno = {45},
 numpages = {12},
 url = {http://doi.acm.org/10.1145/1798596.1798598},
 doi = {10.1145/1798596.1798598},
 acmid = {1798598},
 publisher = {ACM},
 address = {New York, NY, USA},
}

@article{shorthand2,
  title={Shorthand universal cycles for permutations},
  author={Holroyd, Alexander E and Ruskey, Frank and Williams, Aaron},
  journal={Algorithmica},
  volume={64},
  number={2},
  pages={215--245},
  year={2012},
  publisher={Springer}
}

@ARTICLE{karyframework, 
author={D. {Gabric} and J. {Sawada} and A. {Williams} and D. {Wong}}, 
journal={IEEE Transactions on Information Theory},
title={A Successor Rule Framework for Constructing $k$ -Ary de {B}ruijn Sequences and Universal Cycles}, 
year={2020}, volume={66}, number={1}, 
pages={679-687}}

@book {knuth,
    AUTHOR = {Knuth, D. E.},
     TITLE = {The Art of Computer Programming. {V}ol. 4{A}. {C}ombinatorial Algorithms. {P}art 1},
 PUBLISHER = {Addison-Wesley, Upper Saddle River, NJ},
      YEAR = {2011},
     PAGES = {xv+883},
      ISBN = {978-0-201-03804-0; 0-201-03804-8},
   MRCLASS = {68-01 (05A05 05A10 05A17 05A18 05C30 68W40)},
  MRNUMBER = {3444818},
}

@article{ranking,
title = "Practical algorithms to rank necklaces, {L}yndon words, and de {B}ruijn sequences",
journal = "Journal of Discrete Algorithms",
volume = "43",
pages = "95 - 110",
year = "2017",
issn = "1570-8667",
doi = "https://doi.org/10.1016/j.jda.2017.01.003",
url = "http://www.sciencedirect.com/science/article/pii/S1570866717300175",
author = "Joe Sawada and Aaron Williams"
}

@article{kociumaka, 
    author={Kociumaka, Tomasz and Radoszewski, Jakub and Rytter, Wojciech},
  title={Efficient Ranking of {L}yndon Words and Decoding Lexicographically Minimal de {B}ruijn Sequence},
  journal={SIAM J. Discrete Math},
  pages={2027-2046},
      volume={30},
      number={4},
        year={2016}
  }

@article{mitchell-decode,
 author = {Mitchell, C.J. and Etzion, T. and Paterson, K.G.},
 doi = {10.1109/18.532887},
 journal = {IEEE Transactions on Information Theory},
 number = {5},
 pages = {1472-1478},
 title = {A method for constructing decodable de {B}ruijn sequences},
 volume = {42},
 year = {1996}
}

@article{tuliani-decode,
title = "De {B}ruijn sequences with efficient decoding algorithms",
journal = "Discrete Mathematics",
volume = "226",
number = "1",
pages = "313 - 336",
year = "2001",
issn = "0012-365X",
doi = "https://doi.org/10.1016/S0012-365X(00)00117-5",
url = "http://www.sciencedirect.com/science/article/pii/S0012365X00001175",
author = "Jonathan Tuliani"
}

@article{chung,
    title = {Universal cycles for combinatorial structures},
    journal = {Discrete Mathematics},
    volume = {110},
    number = {1},
    pages = {43-59},
    year = {1992},
    issn = {0012-365X},
    doi = {https://doi.org/10.1016/0012-365X(92)90699-G},
    url = {https://www.sciencedirect.com/science/article/pii/0012365X9290699G},
    author = {Fan Chung and Persi Diaconis and Ron Graham}
}

@article{jackson,
    author = {Jackson, Brad},
    year = {1993},
    month = {07},
    pages = {141–150},
    title = {Universal Cycles of $k$-subsets and $k$-permutations},
    volume = {117},
    journal = {Discrete Mathematics},
    doi = {10.1016/0012-365X(93)90330-V}
}

@article{hurlbert,
  title={On universal cycles for $k$-subsets of an $n$-set},
  author={Hurlbert, Glenn},
  journal={SIAM Journal on Discrete Mathematics},
  volume={7},
  number={4},
  pages={598--604},
  year={1994},
  publisher={SIAM}
}

@article{2009research,
  title={Research problems on {G}ray codes and universal cycles},
  author={Jackson, Brad and Stevens, Brett and Hurlbert, Glenn},
  journal={Discrete Mathematics},
  volume={309},
  number={17},
  pages={5341--5348},
  year={2009},
  publisher={Elsevier Science}
}

@inproceedings{sawada2013,
  title={Universal cycles for weight-range binary strings},
  author={Sawada, Joe and Williams, Aaron and Wong, Dennis},
  booktitle={Combinatorial Algorithms: 24th International Workshop, IWOCA 2013, Rouen, France, July 10-12, 2013, Revised Selected Papers 24},
  pages={388--401},
  year={2013},
  organization={Springer}
}

@article{concattree,
  title={Concatenation trees: A framework for efficient universal cycle and de {B}ruijn sequence constructions},
  author={Sawada, Joe and Sears, Jackson and Trautrim, Andrew and Williams, Aaron},
  journal={arXiv preprint arXiv:2308.12405},
  year = 2024
}

@article{fixedweight2,
  title={De {B}ruijn sequences for fixed-weight binary strings},
  author={Ruskey, Frank and Sawada, Joe and Williams, Aaron},
  journal={SIAM Journal on Discrete Mathematics},
  volume={26},
  number={2},
  pages={605--617},
  year={2012},
  publisher={SIAM}
}
\normalsize

\end{document}